\documentclass[12pt]{article}
\usepackage{amsmath}
\usepackage{graphicx,psfrag,epsf}
\usepackage{enumerate}
\usepackage{natbib}
\usepackage{url}

\usepackage[margin = 2.3cm]{geometry}

\usepackage{amssymb, amsthm, ascmac, authblk, bm, booktabs, caption, comment, enumerate, latexsym, lscape, longtable, mathrsfs, mathtools, multirow, natbib, newtxtext, psfrag, setspace, subcaption, thmtools}
\usepackage[stable]{footmisc}

\usepackage{enumitem}

\usepackage{mathtools}
\usepackage{mathrsfs}
\usepackage{color}
\usepackage{pgfplots}
\pgfplotsset{compat=1.16}
\usepackage{here}
\usepackage{multirow}
\usepackage{caption}
\usepackage{authblk}
\usepackage[colorlinks=true, citecolor=teal, linkcolor=teal, urlcolor=teal]{hyperref}
\usepackage{threeparttable}
\usepackage{comment}
\usepackage{stackengine}
\usepackage{tikz}
\usetikzlibrary{arrows.meta}
\usetikzlibrary{bending}

\usepackage{lineno}

\usepackage[capitalize]{cleveref}

\theoremstyle{definition}
\newtheorem{lemma}{Lemma}
\newtheorem{theorem}{Theorem}
\newtheorem{assumption}{Assumption}
\newtheorem*{rmk*}{Remark}

\newcommand{\CIfim}{\text{CI}^{\texttt{fIM}}} 
\newcommand{\CI}{\text{CI}}
\newcommand{\fim}{\texttt{fIM}}
\newcommand{\SE}{\text{SE}}
\newcommand{\Cov}{\text{Cov}}
\newcommand{\Var}{\text{Var}}

\makeatletter
\def\blfootnote#1{%
  \begingroup
  \renewcommand\thefootnote{}\footnote{#1}%
  \addtocounter{footnote}{-1}%
  \endgroup
}
\makeatother

\title{Uniform Confidence Bands for Infinite-Dimensional Partially Identified Parameters}
\author[$\sharp$]{Shunsuke Imai}
\author[$\flat$]{Yuta Okamoto}
\affil[$\sharp$]{Graduate School of Economics, Kyoto University}
\affil[$\flat$]{Graduate School of Economics, Hitotsubashi University}

\begin{document}

\baselineskip=18pt

\maketitle
\blfootnote{${}^{\sharp}$\href{mailto:imai.shunsuke.57n@st.kyoto-u.ac.jp}{imai.shunsuke.57n@st.kyoto-u.ac.jp}. ${}^{\flat}$\href{mailto:yuta.okamoto@r.hit-u.ac.jp}{yuta.okamoto@r.hit-u.ac.jp}. 
First version: \today.}

\begin{abstract}
    Infinite-dimensional parameters are ubiquitous in empirical economics.
    This paper develops an Imbens--Manski--Stoye type confidence band for infinite-dimensional partially identified parameters. 
    In particular, we propose multiplier bootstrap-based construction of a uniform confidence band.
    By employing approximation theorems for suprema of non-centered empirical processes indexed by possibly non-Donsker classes \citep{chernozhukov2016empirical}, we confirm the uniform validity of the proposed procedure.
\end{abstract}

{\textbf{Keywords:} Confidence band, partial identification, uniform inference}

\vspace{0.2cm}

{\textbf{JEL Classification:} C12, C14, C15}

\newpage
\section{Introduction}\label{sec:intro}
Infinite-dimensional parameters are ubiquitous in empirical economics.
A leading example is the distribution of treatment effects, which describes the entire distribution of individual-level treatment effects rather than a finite-dimensional summary such as an average treatment effect \citep{Heckman_etal:1997}. 
Another prominent example is the conditional average treatment effect function with respect to a covariate $X$, which is widely used to summarize treatment effect heterogeneity in empirical work \citep[e.g.][]{lee2017doubly,fan2022estimation,imai2025doubly}. When $X$ is continuously distributed, this object is infinite-dimensional. 
More generally, conditional expectation and conditional quantile functions are central objects in applied econometrics and are inherently infinite-dimensional when the conditioning variables are continuous.

In many empirically relevant settings, however, such infinite-dimensional objects are not point identified. The distribution of treatment effects is generally only partially identified even in experimental settings \citep{Fan_Park_2010, Firpo_Ridder:2019, Heckman_etal:1997}.
The literature has extended partial identification analysis of the treatment effects distribution to panel data and related settings \citep{Callaway:2021, Fan_Yu:2015}, and has studied how additional restrictions can tighten the resulting bounds \citep{Frandsen_Lefgren:2021}. In these extensions, however, the treatment effects distribution typically remains set identified rather than point identified.
Similarly, the conditional average treatment effect function is only partially identified when standard identifying assumptions, such as unconfoundedness, do not hold exactly \citep{Masten_Poirier:2018}. Conditional expectation and quantile functions are also partially identified when outcomes are interval-valued \citep{Beresteanu_Sasaki:2021, LiMoMoPe21, Manski_Tamer:2002}.
Other infinite-dimensional partially identified objects arise in a variety of applications, including marginal treatment effects under sample selection \citep{Bartalotti_etal:2023}, and distributions of valuations in bargaining or auctions \citep{Freyberger_Larsen:2025}.

For a scalar partially identified parameter, the confidence interval proposed by \cite{Imbens_Manski:2004} has become a standard tool for statistical inference.
The key insight underlying their construction is that the length of the identified interval affects the relevant coverage calculation. 
When the identified interval is nearly degenerate, inference resembles the usual point-identified case and calls for a two-sided critical value. 
By contrast, when the identified interval is large relative to sampling uncertainty, coverage is effectively governed by one-sided deviations of the estimated bounds at the endpoints of the identified set. 
The \citeauthor{Imbens_Manski:2004} construction adapts between these two cases and therefore avoids the unnecessary conservatism of applying a conventional two-sided critical value.

Due to this practical appeal, the \citeauthor{Imbens_Manski:2004} confidence interval has been widely used in applied and theoretical work.
For example, \cite{Lee:2009} uses the \citeauthor{Imbens_Manski:2004} confidence interval in his influential bounding approach under sample selection. 
Subsequent work by \cite{Stoye:2009} further clarifies the coverage properties of the Imbens--Manski confidence interval and develops refinements of the original procedure. More recently, \cite{Frandsen_Pond:2025} extends the Imbens--Manski--Stoye confidence interval to vector-valued partially identified parameters.
However, these procedures are designed for finite-dimensional partially identified parameters and do not directly apply to infinite-dimensional objects such as functions or distributions. 

A related literature develops inference procedures based on moment inequalities. Much of this work also focuses on finite-dimensional parameters, including \cite{Andrews_Soares:2010} and \cite{Rosen08}.
Other contributions allow for conditional, continuum-indexed, or increasingly many moment restrictions, but the object of inference is typically still finite-dimensional or a pointwise nonparametric functional rather than an infinite-dimensional partially identified object; see, for example, \cite{Andrews_Shi:2014} and \cite{Menzel:2014}.
More recently, \cite{CCK19moment} develops inference procedures for many moment inequalities, allowing the number of inequalities to be much larger than the sample size, and their framework accomodates infinite-dimensional parameters under some additional assumptions.
Nevertheless, our proposed inference method is different in that our approach exploits the same endpoint structure that underlies the \citeauthor{Imbens_Manski:2004} confidence interval discussed above, and the width of the identified intervals enters into the coverage calculation.
This gives sharpless of the resulting confidence band likewise \cite{Imbens_Manski:2004}, \cite{Stoye:2009}, and \cite{Frandsen_Pond:2025}.
However, this improvement requires an additional computational step; in effect, the procedure determine which endpoint of the interval is relevant for coverage at each evaluation point, so we can obtain the tightness at the cost of computational burden. 
In this sense, our method is complementary to the approach of \cite{CCK19moment}.
For more comprehensive reviews of inference under partial identification and moment inequalities, see \cite{Canay_Shaikh:2017} and \cite{Molinari:2020}.

In the absence of a general inference procedure for infinite-dimensional partially identified objects, applied work has often relied on pointwise inference.
This paper is intended to fill this gap. 
In particular, we propose multiplier bootstrap-based construction of a uniform confidence band for infinite-dimensional partially identified parameters.
By employing approximation theorems for suprema of non-centered empirical processes indexed by possibly non-Donsker classes \citep{chernozhukov2016empirical}, we confirm the uniform validity of the proposed procedure.

\paragraph{Remark on ongoing work.}
The present version focuses on the multiplier-bootstrap implementation. 
We are currently investigating implementation strayegies that are computationally lighter while preserve coverage accuracy. 
One direction is to develop an analytic critical value as an alternative to the multiplier bootstrap. 
Another is to quantify how fine a grid is required for valid approximation between the supremum over the effectively continuous index set and a maximum over finitely many grid points when computing the critical value. 
Future revisions will establish the theoretical validity of these approaches and compare their finite-sample accuracy and computational cost with those of the current implementation through simulations. 

\paragraph{Setup.} 
Suppose, at a point $x\in\mathcal{X}$, $\theta(x) = \theta(x;F)$ is parameter of interest which is not always point identified but we know that $\theta(x; F) \in [\theta_{l,n}(x;F), \theta_{u,n}(x;F)]$ holds,
where a probability distribution $F$ lies in a set of probability distributions $\mathcal{F}$ that satisfies the assumptions given below. 
We often suppress the depedence of $ \theta(x;F)$ on $F$ if no confusion can arise. 
For each $x\in\mathcal{X}$, we define $\Theta_n(x;F) \coloneqq [\theta_{l,n}(x;F), \theta_{u,n}(x;F)]$.
We define the width of $\Theta_{n}(x)$ as $\Delta_n(x) \coloneqq \theta_{u,n}(x)- \theta_{l,n}(x) $.
We also write $\Theta_n(F) \coloneqq \{\theta  : \mathcal{X}\to \mathbb{R} : \theta_{l,n}(x;F) \le \theta(x;F) \le \theta_{u,n}(x;F) \text{ for all } x\in\mathcal{X}\}$.
For each $t\in\{u,l\}$,  let $\hat{\theta}_{t,n}(x)$ be an estimator of $\theta_{t}(x)$ and  $\sigma_{t,n}^2(x) \coloneqq \Var[\hat{\theta}_{t,n}(x)]/r_{t,n}$, where $r_{t,n}$ is a pointwise convergence rate of the estimator to its expectation.

\paragraph{Overview.}

Our aim is to establish the validity of
\begin{align*}
    \hat{\CI}^\fim(x) \coloneqq \left[ \hat{\theta}_{l,n}(x) - \hat{c}_{n}^\fim \widehat{\SE}_{l,n}(x), \quad \hat{\theta}_{u,n}(x) + \hat{c}_{n}^\fim \widehat{\SE}_{u,n}(x)\right],
\end{align*}
where $\hat{c}^\fim_{n}$ is an empirical critical value of $c_n^\fim$ which satisfies
\begin{align*}
    P_F(\theta(x) \in \CIfim_{n}(x) \text{ for all } x\in\mathcal{X} ) = P_F\left(  A_{n}(x) \cap B_{n}(x) \text{ for all } x\in\mathcal{X} \right),
\end{align*}
where we define the events
\begin{align*}
    & A_n(x) \coloneqq \left\{  \theta(x) \le \hat{\theta}_{u,n} (x) + {c}^{\fim}_{n} \frac{\sigma_u(x)}{r_n^{1/2}} \right\}, \quad B_n(x) \coloneqq \left\{ \hat{\theta}_{l,n}(x) -  {c}^{\fim}_{n} \frac{\sigma_l(x)}{r_n^{1/2}}  \le \theta(x) \right\}.
\end{align*}
In a word, $A_n(x)$ is the event in which the upper endpoint of the confidence band lies above $\theta(x)$ and $B_n(x)$ is the event in which the lower endpoint of the confidence band lies below $\theta(x)$.
After some theoretical investigation, such $c^\fim_{n}$ can be, in turn, identified as the solution of
\begin{align*}
    \inf_{\bm{v}\in\{0,1\}^\mathcal{X}} P_F\left[ \sup_{(t,x)\in\{u,l\}\times\mathcal{X}} (\mathbb{G}_n(\tilde\psi_{t,x}) - B_n(t,x;\bm{v})) \le c_{n}^\fim \right] = 1-\alpha,
\end{align*}
where $\bm{v} \coloneqq \{v(x):x\in\mathcal{X}\}$, $\mathbb{G}_n(\tilde\psi_{t,x}) - B_n(t,x;\bm{v})$ is a non-centered empirical process.
Detail definitions of $\mathbb{G}_n(\tilde\psi_{t,x})$ and $B_n(t,x;\bm{v})$ are presented later. 
Given this representation, we can theoretically validate the critical value building on the approximation theorems for suprema of non-centered empirical processes by \cite{chernozhukov2016empirical}.

\paragraph{Paper Organization.}

The rest of the paper is organized as follows. 
In \cref{subsec:gauss}, we prove the uniform validity of the Gaussian approximation.
\cref{subsec:implementation} presents the implement method of the critical value and its uniform validity.
The proofs of theoretical results are in \cref{sec:proof}.

\section{Gaussian Approximation and  Valid Bootstrap Procedure}

\subsection{Gaussian Approximation} \label{subsec:gauss}

\cref{as:assumption1}(i) and (ii) are the uniform counterparts of Assumption 1(i), (ii) in \cite{Stoye:2009} and Assumption 1 in \cite{Frandsen_Pond:2025}.
\cref{as:assumption1}(iii) is that of the assumption in the statement of Lemma 3 in \cite{Stoye:2009}.
\begin{assumption} \label{as:assumption1} \mbox{}
\begin{itemize}
    \item[(i)] 
        For each $t\in\{u,l\}$ and $\varepsilon>0$, there are estimators $\hat{\theta}_{t,n}(x)$ which admit the asympotically linear approximations uniformly over $x\in \mathcal{X}$ and $F\in\mathcal{F}$ as follows
        \begin{align*}
            & r_n^{1/2}\{\hat{\theta}_{t,n}(x) - \theta_{t,n}(x)\} \coloneqq \frac{1}{\sqrt{n}}\sum_{i=1}^n \psi_{t,n}(X_i;x) + \text{Bias}_{t,n}(x) + o_{P_F}(n^{-\varepsilon}),
        \end{align*}
        where $\{\psi_{t,n}(X_i;x)\}_{i=1}^n$ is a centered and possibly sample size dependent i.i.d.~random variables, $r_n^{1/2}$ is a sample-size dependent sequence which indicates the convergence rate of the estimator and the bias term is uniformly and strictly dominated by the stochastic term.
    \item[(ii)] 
        For all $F\in\mathcal{F}$ and $n\ge 1$, $\inf_{(t,x)\in \{u,l\}\times\mathcal{X}}\sigma_{t,n}(x)$ is bounded away from below by a finite positive constant and $\sup_{(t,x)\in\{u,l\}\times\mathcal{X}} \sigma_{t,n}(x)$ is bounded from above by a finite positive constant. 
        Also, for all $F\in\mathcal{F}$ and $n\ge 1$,  $\sup_{x\in\mathcal{X}}\Delta_n(x)<\infty$.
        In addition, let $\hat{\sigma}_{t,n}(x)$ be an estimator for the standard error and satisfies $\sup_{(t,x)\in \{u,l\}\times \mathcal{X}} |\hat{\sigma}_{t,n}(x)/\sigma_{t}(x) - 1| = o_p(1)$.
    \item[(iii)]
        For all $F\in\mathcal{F}$ and $n\ge 1$, $P_F\{\hat{\theta}_{u,n}(x) \geq \hat{\theta}_{l,n}(x) \text{ for all } x\in\mathcal{X}\}=1$.
\end{itemize}
\end{assumption}

The uniform counterpart of Assumption 3 in \cite{Stoye:2009} and Assumption 2 in \cite{Frandsen_Pond:2025} is given by as follows.
\begin{assumption} \label{as:assumption2} \mbox{}
\begin{itemize}
    \item[(i)] There exists a sequence $\{k_n\}$ such that $k_n \to 0$, $k_n \cdot r_n^{1/2}\to \infty$, and $\sup_{x\in\underline{\mathcal{X}}} r_n^{1/2}|\hat{\Delta}_n(x) - \Delta_n(x)| \xrightarrow{P_F} 0$ uniformly in $F\in\mathcal{F}$ over $\underline{\mathcal{X}} \coloneqq \{x : \Delta_n(x) \le k_n\}$.
    \item[(ii)] $\Delta_n(x)$ satisfies $\sqrt{r_n}\Delta_n(x) \to \Delta_\star(x)\in[0,+\infty]$ for all $x\in\mathcal{X}$ and $F\in\mathcal{F}$.
\end{itemize}

\end{assumption}

We assume that the empirical critical value is consistent. 
We present how to implement such a critical value in the following section.
\begin{assumption} \label{as:assumption3}
The empirical critical value $\hat{c}_n^\fim$ is consistent to $c_n^\fim$.
\end{assumption}

In order to establish a Gaussian approximation, we define some notations.
Let
\begin{align*}
    s_u \coloneqq (s_{u,1}, s_{u,2}) = (1,0) , \quad s_l \coloneqq (s_{l,1}, s_{l,2}) = (0,1) ,
\end{align*}
and
\begin{align*}
    S_{u,n}(x) \coloneqq r_{n}^{1/2} \frac{\hat{\theta}_{u,n}(x) - \theta_{u,n}(x)}{\sigma_u(x)} , \quad 
    S_{l,n}(x) \coloneqq r_{n}^{1/2} \frac{\hat{\theta}_{l,n}(x) - \theta_{l,n}(x)}{\sigma_l(x)} .
\end{align*}
For each  $(t,x) \in \{u,l\} \times \mathcal{X}$, define
\begin{align*}
    G_{n}(t,x) & \coloneqq s_t
    \begin{bmatrix}
        - S_{u,n}(x) , S_{l,n}(x)
    \end{bmatrix}^\top
    = -s_{t,1} S_{u,n}(x)  + s_{t,2} S_{l,n}(x),
\end{align*}
and
\begin{align*}
    B_{n}(t,x) &\coloneqq s_t 
    \begin{bmatrix}
        r_{n}^{1/2}\frac{(1-v(x))\Delta_{n}(x)}{\sigma_{u}(x)} , r_{n}^{1/2} \frac{v(x) \Delta_{n}(x)}{\sigma_l(x)} 
    \end{bmatrix}^\top \\
    &= s_{t,1} \left(r_{n}^{1/2}\frac{(1-v(x))\Delta_{n}(x)}{\sigma_{u}(x)}\right) + s_{t,2} \left( r_{n}^{1/2} \frac{v(x) \Delta_{n}(x)}{\sigma_l(x)} \right).
\end{align*}

As mentioned in \cref{sec:intro}, we build on the results in \cite{chernozhukov2016empirical}. 
In order to utilize their results, 
define $\tilde\Psi \coloneqq \{ \tilde\psi_{t,x} : (t,x)\in \{u,l\}\times \mathcal{X} \}$, with $\tilde{\psi}_{t,x}(\cdot) \coloneqq -s_{t,1} \psi_{u,n}(\cdot,x)/\sigma_u(x) + s_{t,2}  \psi_{l,n}(\cdot,x)/\sigma_l(x)$. 
In some case, we define $B : \tilde\Psi \mapsto \mathbb{R}$ so that $B(\cdot) \coloneqq \inf\{ B_n(t,x) : \tilde\psi_{t,x} = \cdot\}$. 
For each $F\in\mathcal{F}$ and $\eta >0$, let $N_B(\eta ; F)$ be the minimal integer $N$ such that there exist $\tilde\psi_{t_1,x_1}, \dots, \tilde\psi_{t_N,x_N}\in \tilde\Psi$ with the property that, for every $\tilde\psi_{t,x}\in\tilde\Psi$, there exists $1\le j \le N$ with $|B(\tilde\psi_{t,x}) - B(\tilde\psi_{t_j,x_j})| < \eta$.
Also define $K_n(F) \coloneqq \log N_B(\eta;F) + v\{\log n \vee \log(Ab(F)/\sigma(F))\}$.
We introduce following \cref{as:assumption4}(i)-(iii), which corespond assumptions (A), (B) and (C) in \cite{chernozhukov2016empirical}.
\cref{as:assumption4}(iv) is an additional assumption for Gaussian approximation by \cite{chernozhukov2016empirical} (cf. the statement of their Theorem 2.1).
\begin{assumption} \label{as:assumption4}\mbox{} 
\begin{itemize}
    \item[(i)] For all $F\in\mathcal{F}$, there exists a countable subset $\mathcal{G}_{\tilde\Psi}$ of $\tilde\Psi$ such that for any $\tilde\psi_{t,x}\in \tilde\Psi$, there exist sequences  $g_{m}\in \mathcal{G}_{\tilde\Psi}$ with $g_{m} \to \tilde\psi_{t,x}$ pointwise and $B(g_{m}) \to B(\tilde\psi_{t,x})$.
    \item[(ii)] For all $F\in\mathcal{F}$, the class of functions $\tilde\Psi$ is VC-type with a measurable envelope $\bar{\Psi}$ and constants $A\ge e$ and $v\ge 1$. 
    \item[(iii)] For each $F\in\mathcal{F}$, there exist constants $b(F) \ge \sigma(F) >0$ and $q\in[4,\infty)$ such that $\sup_{(t,x)\in \{u,l\}\times\mathcal{X}} \mathbb{E}_F[|\tilde\psi_{t,x}(X_i)|^k] \le \sigma(F)^2 b(F)^{k-2}$ for $k=2,3,4$ amd $ \mathbb{E}_F[|\bar\Psi(X_i)|^q] \le b(F)^q$.
    \item[(iv)] For all $F\in\mathcal{F}$, $K_n(F)^3 \le n$.
\end{itemize}
\end{assumption}

The following assumption makes the Gaussian approximation valid uniformly over $F\in\mathcal{F}$.
\begin{assumption} \label{as:assumption5}\mbox{}
\begin{itemize}
    \item[(i)]$\inf_{F\in\mathcal{F}}\inf_{\tilde\psi_{t,x}\in\tilde\Psi} \Var_F\{G(\tilde\psi_{t,x})\} > 0$. 
    \item[(ii)] $\sup_{F\in\mathcal{F}} (\inf_{\tilde\psi_{t,x}\in\tilde\Psi} \Var_F\{G(\tilde\psi_{t,x})\} )^{-1/2} r_1(F)\{1 + \sqrt{K_n(F)}\} \to 0$ as $n\to\infty$, where $r_1(F) \coloneqq  C_1\{\eta + \delta_n^{(1)}(F)\}$ with
    \begin{align*}
        & \delta_n^{(1)}(F) \coloneqq \frac{b(F)K_n(F)}{\gamma^{1/q }n^{1/2 - 1/q}} + \frac{\{b(F) \sigma^2(F) K_n^2(F)\}^{1/3}}{\gamma_n^{1/3} n^{1/6}},
    \end{align*}
    where $\gamma_n\in(0,1)$ is a sample size-dependent sequence such that $\gamma_n \to 0$.
\end{itemize}
\end{assumption}

\begin{theorem} \label{thm:gaussian_approximation}
Under \cref{as:assumption1}-\ref{as:assumption5}, it holds that
\begin{align*}
    \liminf_{n\to\infty} \inf_{F\in\mathcal{F}}\inf_{\theta\in \Theta_n(F)}  P_F(\theta(x) \in \CIfim_n(x) \text{ for all } x\in\mathcal{X}) = 1-\alpha.
\end{align*}
\end{theorem}

\subsection{Implementation of Critical Value} \label{subsec:implementation}

In this section, we consider the implementation of the critical value via multiplier bootstrap building on \cite{chernozhukov2016empirical}.
For each $b= 1, \dots B$, we generate a set of i.i.d.~bootstrap weights $\{w_i^{\star,b}\}_{i=1}^n$ independently of the origina data.
Following \cite{chernozhukov2016empirical}, we suppose that  $\{w_i^{\star,b}\}_{i=1}^n$ follows the standard normal distribution.
In each bootstrap iteration, we compute
\begin{align*}
    \mathbb{G}_n^{\star,b}(\tilde\psi^{\text{s.a.}}_{t,x}) \coloneqq \frac{1}{\sqrt{n}}\sum_{i=1}^n w_i^{\star,b}  \left(\tilde{\psi}^{\text{s.a.}}_{t,x}(X_i)  - \frac{1}{n}\sum_{i=1}^n \tilde{\psi}^{\text{s.a.}}_{t,x}(X_i)  \right),
\end{align*}
where $\tilde{\psi}^{\text{s.a.}}_{t,x}$ is a sample analogue of $\tilde\psi_{t,x}$ (s.a. indicates ``sample analogue'').
Set
\begin{align*}
    \hat T^{\star,b}_n(v) \coloneqq \sup_{(t,x)\in\{u,l\}\times\mathcal{X}} \left( \mathbb{G}_n^{\star,b}(\tilde{\psi}^{\text{s.a.}}_{t,x}) - \hat{B}_n(t,x;v) \right),
\end{align*}
where $\hat{B}_n(t,x;v)$ is a sample analogue of $B_n(t,x;v)$, and
\begin{align*}
    \hat q^{\star}_n(v) \coloneqq \text{ $(1-\alpha)$-quantile of $\{\hat T^{\star,b}_n(v)\}_{b=1}^B$}.
\end{align*}
Then, the critical value is defined as
\begin{align}
    \hat{c}^\fim_n \coloneqq \sup_{v\in\{0,1\}^\mathcal{X}}  \hat q^{\star}_n(v). \label{eq:cv_mb}
\end{align}

\begin{assumption}  \mbox{}\label{as:assumption6}
\begin{itemize}
    \item[(i)] The sample analogues $\tilde{\psi}^{\text{s.a.}}_{t,x}$ and $\hat{B}_n(t,x;v)$ are consistent uniformly over $F\in\mathcal{F}$, $(t,x)\in\{u,l\}\times\mathcal{X}$ and $(t,x,v)\in\{u,l\}\times\mathcal{X}\times \bm{v}$.
    \item[(ii)] $\sup_{F\in\mathcal{F}} (\inf_{\tilde\psi_{t,x}\in\tilde\Psi} \Var_F\{G(\tilde\psi_{t,x})\} )^{-1/2} r_1^\star(F)\{1 + \sqrt{K_n(F)}\} \to 0$ as $n\to\infty$, where $r_1^\star(F) \coloneqq  C_3\{\eta + \delta_n^{(2)}(F)\}$ with
        \begin{align*}
            \delta_n^{(2)}(F) \coloneqq \frac{b(F) K_n(F)}{\gamma_n^{1 + 1/q} n^{1/2 -1/q}} + \frac{\{b(F)\sigma(F)K_n(F)^{3/2}\}^{1/2}}{\gamma_n^{1 + 1/q} n^{1/4}}.
        \end{align*}
    \item[(iii)]  For all $F\in\mathcal{F}$, $K_n(F) \le n$.
\end{itemize}

\end{assumption}

\begin{theorem} \label{thm:vaild_mb}
Under \cref{as:assumption1}, \ref{as:assumption4}, \ref{as:assumption5} and \ref{as:assumption6}, the critical value defined as \cref{eq:cv_mb} satisfies \cref{as:assumption3}.
\end{theorem}

\section{Proofs} \label{sec:proof}

\subsection{Proofs for \cref{subsec:gauss}}

Before the proof of main theorem, we provide the following lemmas.
\begin{lemma} \label{lem:uniform_Gauss}
Under \cref{as:assumption1}, \ref{as:assumption4} and \ref{as:assumption5}, it holds that
\begin{align*}
    \sup_{F\in\mathcal{F}} \left| P_F(\theta(x) \in \CIfim_n(x) \text{ for all } x\in\mathcal{X} ) -  P_F\left[ \sup_{\tilde\psi_{t,x}\in\tilde\Psi} (G(\tilde\psi_{t,x}) - B(\tilde\psi_{t,x})) \le c_{n}^\fim \right]\right| = o(1),
\end{align*} 
where the covariance function of the centered Gaussian process $\{G(\tilde \psi_{t,x}) : (t,x)\in \{u,l\}\times\mathcal{X}\}$ is 
\begin{align*}
    \Cov\{G(\tilde\psi_{t,x}), G(\tilde\psi_{t',x'})\} = s_t
    \begin{bmatrix}
        \Sigma_{uu}(x,x') & -\Sigma_{ul}(x,x') \\ -\Sigma_{lu}(x,x') & \Sigma_{ll}(x,x')
    \end{bmatrix} s_{t'}^\top
\end{align*}
with $\Sigma_{tt'}(x,x') \coloneqq \mathbb{E}_F[\psi_t(X;x)\psi_{t'}(X;x')]/\sigma_t(x)\sigma_{t'}(x')$.
\end{lemma}

\begin{proof}
The coverage probability in terms of $\theta(x)$ , $\CIfim_{n}(x)$ is defined so that it holds that
\begin{align*}
    P_F(\theta(x) \in \CIfim_{n}(x) \text{ for all } x\in\mathcal{X} ) = P_F\left( \tilde A_{n}(x) \cap\tilde B_{n}(x) \text{ for all } x\in\mathcal{X} \right),
\end{align*}
where we define the events
\begin{align*}
    & \tilde{A}_{n}(x) \coloneqq \left\{ -{c}^{\fim}_{n} - r_{n}^{1/2}\frac{\theta_{u,n}(x) - \theta(x)}{\sigma_{u}(x)} \le r_{n}^{1/2} \frac{\hat{\theta}_{u,n}(x) - \theta_{u,n}(x)}{\sigma_u(x)}\right\}, \\
    & \tilde{B}_{n}(x) \coloneqq \left\{ r_{n}^{1/2} \frac{\hat{\theta}_{l,n}(x) - \theta_{l,n}(x)}{\sigma_l(x)} \le {c}_{n}^\fim + r_{n}^{1/2} \frac{\theta(x) - \theta_{l,n}(x)}{\sigma_l(x)} \right\}.
\end{align*}
Also, we parametrize $\theta(x) = \theta_{l,n}(x) + v(x)\Delta_{n}(x) = \theta_{u,n}(x) - (1-v(x))\Delta_{n}(x)$ without loss of generarity. 
Under this parametrization, we can see that the events $\tilde{A}_{n}(x)$ and $ \tilde{B}_{n}(x)$ satisfy
\begin{align*}
    \tilde{A}_{n}(x) 
    & = \left\{ -{c}^{\fim}_{n} - r_{n}^{1/2}\frac{(1-v(x))\Delta_{n}(x)}{\sigma_{u}(x)} \le r_{n}^{1/2} \frac{\hat{\theta}_{u,n}(x) - \theta_{u,n}(x)}{\sigma_u(x)} \right\} \\
    & = \left\{- r_{n}^{1/2} \frac{\hat{\theta}_{u,n}(x) - \theta_{u,n}(x)}{\sigma_u(x)}  - r_{n}^{1/2}\frac{(1-v(x))\Delta_{n}(x)}{\sigma_{u}(x)} \le c_{n}^\fim \right\},
\end{align*}
and
\begin{align*}
    \tilde{B}_{n}(x) &= \left\{ r_{n}^{1/2} \frac{\hat{\theta}_{l,n}(x) - \theta_{l,n}(x)}{\sigma_l(x)} \le {c}_{n}^\fim + r_{n}^{1/2} \frac{v(x) \Delta_{n}(x)}{\sigma_l(x)} \right\} \\
    &= \left\{  r_{n}^{1/2} \frac{\hat{\theta}_{l,n}(x) - \theta_{l,n}(x)}{\sigma_l(x)}  -  r_{n}^{1/2} \frac{v(x) \Delta_{n}(x)}{\sigma_l(x)} \le c^\fim_{n}\right\}.
\end{align*}
Therefore, it holds that
\begin{align*}
    & P_F(\theta(x) \in \CIfim_{n}(x) \text{ for all } x\in\mathcal{X} ) \\
    & = P_F\left[ \sup_{x\in\mathcal{X}} \left( - S_{u,n}(x)  - r_{n}^{1/2}\frac{(1-v(x))\Delta_{n}(x)}{\sigma_{u}(x)}  \right)\le c^\fim_{n}, \sup_{x\in\mathcal{X}} \left( S_{l,n}(x)  -  r_{n}^{1/2} \frac{v(x) \Delta_{n}(x)}{\sigma_l(x)} \right) \le c^\fim_{n}\right] \\
    & = P_F\left[  \sup_{x\in\mathcal{X}} \left( - S_{u,n}(x)  - r_{n}^{1/2}\frac{(1-v(x))\Delta_{n}(x)}{\sigma_{u}(x)}  \right) \vee \sup_{x\in\mathcal{X}} \left( S_{l,n}(x)  -  r_{n}^{1/2} \frac{v(x) \Delta_{n}(x)}{\sigma_l(x)} \right)\le c^\fim_{n}\right] .
\end{align*}
Then, 
\begin{align*}
    P_F(\theta(x) \in \CIfim_{n}(x) \text{ for all } x\in\mathcal{X} ) = P_F\left[ \sup_{(t,x)\in\{u,l\}\times \mathcal{X}} \left(  G_n(t,x) - B_n(t,x)\right) \le c_{n}^\fim\right].
\end{align*}
Since, under Assumption \ref{as:assumption1}, it holds that
\begin{align*}
    \sup_{F\in\mathcal{F}} P_F\left(  \sup_{(t,x)\in\{u,l\}\times\mathcal{X}}\left| G_n(t,x) - \mathbb{G}_n(\tilde\psi_{t,x})\right|  \ge \varepsilon \right) \to 0, \quad \forall \varepsilon>0,
\end{align*}
we can regard $G_n(t,x)$ as the empirical process $\mathbb{G}_n(\tilde\psi_{t,x})$ asymptotically, therefore we can approximate the uniform coverage probability of interest as the distribution function of the supremum of the non-centered empirical process;
\begin{align*}
    \sup_{F\in\mathcal{F}} \left| P_F\left(\theta(x) \in \CIfim_n(x) \text{ for all } x\in\mathcal{X} \right) - P_F\left[ \sup_{(t,x)\in\{u,l\}\times \mathcal{X}} \left(  \mathbb{G}_n(\tilde\psi_{t,x}) - B(\tilde\psi_{t,x})\right) \le c^\fim_{n} \right] \right| = o(1).
\end{align*}
For each $F\in\mathcal{F}$ and every $\gamma_n\in(0,1)$, Theorem 2.1 in \citep{chernozhukov2016empirical} states that, under \cref{as:assumption4}, it holds that 
\begin{align*}
    P_F\left( \left|\sup_{\tilde\psi_{t,x}\in\tilde\Psi}\left(  \mathbb{G}_n(\tilde\psi_{t,x}) - B(\tilde\psi_{t,x})\right) - \sup_{\tilde\psi_{t,x}\in\tilde\Psi}(G(\tilde\psi_{t,x}) - B(\tilde\psi_{t,x}))\right| > r_1(F)\right) \le r_2 
\end{align*}
with $ r_2 \coloneqq C_2(\gamma_n + n^{-1})$.
In conjunction with \cref{as:assumption4}(ii) and Lemma 2.1 and Lemma 2.2 in \cite{chernozhukov2016empirical} (with appropriate choice of $r$ and $\delta$  so that $r\delta \lesssim r_1(F)$ and $\phi(\delta) \lesssim r_1(F)$ in their notation), 
we have
\begin{align}
    & \sup_{F\in\mathcal{F}}\sup_{z\in\mathbb{R}}  \left|P_F\left[\sup_{\tilde\psi_{t,x}\in\tilde\Psi} \left(  \mathbb{G}_n(\tilde\psi_{t,x}) - B(\tilde\psi_{t,x})\right) \le z \right] -  P_F\left[ \sup_{\tilde\psi_{t,x}\in\tilde\Psi} (G(\tilde\psi_{t,x}) - B(\tilde\psi_{t,x})) \le z \right]\right| \nonumber\\
    & \le \sup_{F\in\mathcal{F}}\sup_{z\in\mathbb{R}}  \left| P_F\left[ \left|\sup_{\tilde\psi_{t,x}\in\tilde\Psi} (G(\tilde\psi_{t,x}) - B(\tilde\psi_{t,x})) - z\right| \le r_1(F) \right]\right| + r_2\nonumber \\
    & \lesssim \sup_{F\in\mathcal{F}} \frac{2}{\sqrt{\inf_{\tilde\psi_{t,x}\in\tilde\Psi} \Var_F\{G(\tilde\psi_{t,x})\}}} r_1(F) \left(1 + \sqrt{K_n(F)}\right)  + r_2.  \nonumber
\end{align}
Therefore
\begin{align}
    \sup_{F\in\mathcal{F}} \left| P_F(\theta(x) \in \CIfim_n(x) \text{ for all } x\in\mathcal{X} ) -  P_F\left[ \sup_{\tilde\psi_{t,x}\in\tilde\Psi} (G(\tilde\psi_{t,x}) - B(\tilde\psi_{t,x})) \le c_{n}^\fim \right]\right| = o(1). \nonumber
\end{align} 
\end{proof}

\begin{lemma} \label{lem:lem2-uni}
\cref{as:assumption2}(i) is satisfied under \cref{as:assumption1}, \ref{as:assumption4} and \ref{as:assumption5}. 
\end{lemma}

\begin{proof}
In the following proof, we denote ${\sigma}^2_{l}(x) \coloneqq r_n \tilde{\sigma}^2_l(x)$, where $r_n$ is a convergence rate of some estimator of $\theta_l(x)$ and $\tilde{\sigma}_l(x)$ is an appropriate constant.

Fisrt, we will show $\sup_{F\in\mathcal{F}}\sup_{x\in\underline{\mathcal{X}}_n(F)} \tilde{\sigma}_\Delta(x;F) \to 0$ by contradiction, where $\underline{\mathcal{X}}_{k_n}(F) \coloneqq \{x\in\mathcal{X}: \Delta_n(x;F) \le k_n\}$.
The sequence $k_n$ is chosen below so that $k_n \to 0$, $\sqrt{r}_nk_n\to\infty$.
The proof of $\sup_{F\in\mathcal{F}}\sup_{x\in\underline{\mathcal{X}}_{k_n}(F)} \tilde{\sigma}_\Delta(x;F) \to 0$ is almost same as that of conterpart result in \cite{Stoye:2009}. 
We assume $\sup_{F\in\mathcal{F}}\sup_{x\in\underline{\mathcal{X}}_{k_n}(F)} \tilde{\sigma}_\Delta(x;F) $ does not converge to $0$. 
This implies that there exist $\eta>0$, a subsequence $\{n_m\}_{m\ge 1}$, distributions $F_m\in\mathcal{F}$ and points $x_m\in \underline{\mathcal{X}}_{k_{n_m}}(F)$ such that $\tilde{\sigma}_\Delta(x_m;F_m) \ge \eta$ for all $m \ge 1$.
It is sufficient to show that $P_{F_m}[\hat{\Delta}_{n_m}(x_m) < 0] > 0$ holds  for all sufficiently large $m$. 
Indeed, this implies $P_{F_m}[\hat{\theta}_{u,n_m}(x_m) < \hat{\theta}_{l,n_m}(x_m)]>0$, which contradicts \cref{as:assumption1}(iii).
\cref{as:assumption1} and the standard Berry-Esseen's theorem give
\begin{align}
    \sup_{F\in\mathcal{F},x\in\mathcal{X},t\in\mathbb{R}} \left| P_F\left( r_n^{1/2}\{\hat{\Delta}_n(x) - \Delta_n(x)\} \le t\right)  - \Phi\left(\frac{t}{\tilde{\sigma}_{\Delta}(x;F)}\right) \right| \le \varepsilon_n,\label{eq:app_point_dist} 
\end{align}
with $\varepsilon_n\to 0$.
As stated in the proof of Lemma 3 in \cite{Stoye:2009}, we can fix a non-positive sequence $\delta_n \to -\infty$ such that $-\delta_m/\sqrt{r_n} \to 0$ and $\Phi(\gamma\delta_n) \gg \varepsilon_n$  for any fixed $\gamma>0$.
Choose $k_n = -\delta_n / (2\sqrt{r_n})$, then $k_n \to 0$ and $\sqrt{r_n}k_n = -\delta_n/2 \to \infty$.
Then, as proof of Lemma 3 in \cite{Stoye:2009}, from \eqref{eq:app_point_dist}, we have 
\begin{align*}
    &P_{F_m}[\hat{\theta}_{u,n_m}(x_m) < \hat{\theta}_{l,n_m}(x_m)]\\
    & = P_{F_m}[r_{n_m}^{1/2}\{\hat{\Delta}_{n_m}(x_m) - \Delta_{n_m}(x_m)\} < -r_{n_m}^{1/2} \Delta_{n_m}(x_m)] \\
    & \ge \Phi(-r_{n_m}^{1/2}\Delta_{n_m}(x_m)/\tilde{\sigma}_{\Delta}(x_m;F_m)) - \varepsilon_{n_m} \\
    & \ge \Phi(\delta_{n_m}/\tilde{\sigma}_{\Delta}(x_m;F_m)) - \varepsilon_{n_m},
\end{align*}
which is strictly positive for $n$ large enough. 
Moreover, since $\tilde{\sigma}_\Delta(x_m;F_m) \ge \eta$ and $\delta_{n_m} < 0$, it follows that $P_{F_m}[\hat{\theta}_{u,n_m}(x_m) < \hat{\theta}_{l,n_m}(x_m)] \ge \Phi\left( \delta_{n_m}/\eta\right) - \varepsilon_{n_m}$.
By the choice of $\delta_n$ and $\varepsilon_n$, $\Phi\left( \delta_{n_m}/\eta\right) \gg \varepsilon_{n_m}$ for all sufficiently large $m$.
Therefore
\begin{align*}
    P_{F_{m}}[\hat{\theta}_{u,n_m}(x_m) < \hat{\theta}_{l,n_m}(x_m)]>0.
\end{align*}
This contradicts \cref{as:assumption1}(iii), thus $\sup_{F\in\mathcal{F}}\sup_{x\in\underline{\mathcal{X}}_{k_n}(F)} \tilde{\sigma}_\Delta(x;F) \to 0$.

By applying the similar uniform Gaussian approximation argument as in the proof of \cref{lem:uniform_Gauss}
to the empirical process indexed by $\{\pm(\psi_{u,n}(\cdot;x)-\psi_{l,n}(\cdot;x)) : x\in\underline{\mathcal{X}}_{k_n}(F)\}$, 
we obtain
\begin{align}
    \sup_{F\in\mathcal{F},t\in\mathbb{R}} \left| P_F\left( \sup_{x\in\underline{\mathcal{X}}} r_n^{1/2}|\hat{\Delta}_n(x) - \Delta_n(x)| \le t\right)  - P_F\left(\sup_{x\in\underline{\mathcal{X}}} |\tilde{\sigma}_{\Delta}(x)G(x)| \le t \right)\right| \le \varepsilon_n, \label{eq:app_uni_dist}
\end{align}
with a sequence $\varepsilon_n \to 0$ and the standard Gaussian process $G(x)$. 
Since we have shown \\$\sup_{F\in\mathcal{F}}\sup_{x\in\underline{\mathcal{X}}_{k_n}(F)} \tilde{\sigma}_\Delta(x;F) \to 0$, in conjunction with \eqref{eq:app_uni_dist}, it holds that
\begin{align*}
    & \sup_{F\in \mathcal{F} } P_F\left( \sup_{x\in\underline{\mathcal{X}}_{k_n}(F)}  r_n^{1/2}\left| \hat{\Delta}_n(x) - \Delta_n(x) \right| > t\right) \\
    & \le \sup_{F\in \mathcal{F}} P_F\left( \sup_{x\in\underline{\mathcal{X}}_{k_n}(F)}  \tilde{\sigma}_\Delta(x)G(x) > t \right) + \varepsilon_n \\
    & \le \sup_{F\in \mathcal{F}} P_F\left( \sup_{x\in\underline{\mathcal{X}}_{k_n}(F)}  \left|\tilde{\sigma}_\Delta(x)\right| \sup_{x\in\underline{\mathcal{X}}_{k_n}(F)}\left|G(x)\right| > t\right) + \varepsilon_n \\
    & = \sup_{F\in \mathcal{F} } P_F\left(  \sup_{x\in\underline{\mathcal{X}}_{k_n}(F)}\left|G(x)\right| > \frac{t}{\sup_{x\in\underline{\mathcal{X}}_{k_n}(F)}  \left|\tilde{\sigma}_\Delta(x)\right|} \right) + \varepsilon_n \\
    & \le \frac{\sup_{F\in \mathcal{F} }\sup_{x\in\underline{\mathcal{X}}_{k_n}(F)}|\tilde{\sigma}_{\Delta}(x)|}{t} \cdot \sup_{F\in \mathcal{F} }  \mathbb{E}_{F}\left[ \sup_{x\in\underline{\mathcal{X}}_{k_n}(F)}\left|G(x)\right| \right] + \varepsilon_n.
\end{align*}
Under \cref{as:assumption4} and \cref{as:assumption5}, it holds that $\mathbb{E}_{F}\left[ \sup_{x\in\underline{\mathcal{X}}_{k_n}(F)}\left|G(x)\right| \right] < \infty$ uniformly over $\mathcal{F}$.
Therefore, $\sup_{F\in\mathcal{F}}\sup_{x\in\underline{\mathcal{X}}_{k_n}(F)}|\tilde{\sigma}_{\Delta}(x)|\to 0$ and $\varepsilon_n\to 0$ imply $\sup_{x\in\underline{\mathcal{X}}_{k_n}(F)}r_n^{1/2}|\hat{\Delta}_n(x) - \Delta_n(x)| \xrightarrow{p} 0$ uniformly in $\mathcal{F}$.

\end{proof}

\begin{proof}[Proof of \cref{thm:gaussian_approximation}]
The proof here generalize the results in \cite{Stoye:2009} and \cite{Frandsen_Pond:2025} to infinite-dimensional setting.
The proof proceeds in line with that of \cite{Frandsen_Pond:2025}, though technical details are different.

First, note that $P_F\left(\theta(x) \in \hat{\CI}^\fim_n(x)\right)$ is a sequence of real numbers. 
By the definition of $\lim \inf$, there exists a subsequence $\{n_m\}$ such that,
\begin{align*}
    & \liminf_{n\to\infty}\inf_{F\in\mathcal{F}}  \inf_{\theta\in\Theta_n(F)} P_F\left(\theta(x) \in \hat{\CI}^\fim_n(x) \text{ for all } x\in\mathcal{X} \right) \\
    & = \lim_{m\to\infty} \inf_{F\in\mathcal{F}}  \inf_{\theta\in\Theta_{n_m}(F)}  P_F\left(\theta(x) \in \hat{\CI}^\fim_{n_m}(x) \text{ for all } x\in\mathcal{X} \right) .
\end{align*}
We work with subsequence $\{n_m\}$. 
Note that the coverage probability of $\hat{\CI}^\fim_{n_m}(x)$ is given by 
\begin{align*}
    P_F\left(\theta(x) \in \hat{\CI}^\fim_{n_m}(x) \text{ for all } x\in\mathcal{X} \right) = P_F\left( \hat{A}_{n_m}(x) \cap \hat{B}_{n_m}(x) \text{ for all } x\in\mathcal{X} \right),
\end{align*}
where we define the events
\begin{align*}
    & \hat{A}_{n_m}(x) \coloneqq \left\{ -\frac{\hat{\sigma}_{u,n_m}(x)}{\sigma_u(x)} \hat{c}^{\fim}_{n_m} - r_{n_m}^{1/2}\frac{\theta_{u,n_m}(x) - \theta(x)}{\sigma_{u}(x)} \le r_{n_m}^{1/2} \frac{\hat{\theta}_{u,n_m}(x) - \theta_{u,n_m}(x)}{\sigma_u(x)}\right\}, \\
    & \hat{B}_{n_m}(x) \coloneqq \left\{ r_{n_m}^{1/2} \frac{\hat{\theta}_{l,n_m}(x) - \theta_{l,n_m}(x)}{\sigma_l(x)} \le \frac{\hat{\sigma}_{l,n_m}(x)}{\sigma_{l}(x)} \hat{c}_{n_m}^\fim + r_{n_m}^{1/2} \frac{\theta(x) - \theta_{l,n_m}(x)}{\sigma_l(x)} \right\}.
\end{align*}

\paragraph{Case 1}
We first consider the case where $\Delta_{n_m}(x) \le k_n$ for all $x\in\mathcal{X}$, where $k_n$ is the sequence in Assumption \ref{as:assumption2}. 
From \cref{as:assumption1}, we have
\begin{align}
    \sup_{(t,x) \in \{u,l\} \times \mathcal{X}} \left| \frac{\hat{\sigma}_{t,n_m}(x)}{\sigma_t(x)} - 1 \right| = o_{P_F}(1). \label{eq:variance_consistency}
\end{align}
Also, from \cref{as:assumption3}, we have $|\hat{c}_{n_m} - c_{n_m}| = o_{P_F}(1)$.
Then, we can see that
\begin{align*}
    & \inf_{F\in\mathcal{F}}  \inf_{\theta\in\Theta_{n_m}(F)} P_F\left( \theta(x) \in \hat{\CI}^\fim_{n_m}(x)  \text{ for all } x\in\mathcal{X}   \right) \\
    & = \inf_{F\in\mathcal{F}}  \inf_{\theta\in\Theta_{n_m}(F)} P_F\left( \hat{A}_{n_m}(x) \cap \hat{B}_{n_m}(x) \text{ for all } x\in\mathcal{X} \right) \\
    & =\inf_{F\in\mathcal{F}}  \inf_{\theta\in\Theta_{n_m}(F)} P_F\left( \tilde{A}_{n_m}(x) \cap \tilde{B}_{n_m}(x) \text{ for all } x\in\mathcal{X} \right) + o(1) \\
    & =  \inf_{F\in\mathcal{F}}  \inf_{\theta\in\Theta_{n_m}(F)} P_F \left[ \sup_{(t,x)\in\{u,l\}\times\mathcal{X}} (G(\tilde\psi_{t,x}) - B_n(t,x)) \le c_{n_m}^\fim \right] + o(1), 
\end{align*}
where 
the first equality follows from the definition of $\CIfim_{n_m} $, $\tilde{A}_{n_m}(x)$ and $\tilde{B}_{n_m}(x)$, 
the second approximation follows from \eqref{eq:variance_consistency} and \cref{as:assumption3}, and
the final approximation follows from the unifrom Gaussian approximation (\cref{lem:uniform_Gauss}).

For the subsequent discussion, note that $B_n(t,x)$ is function with respect to $v$,
\begin{align*}
    B_n(t,x) = s_t \begin{bmatrix}
        r_{n_m}^{1/2}\frac{(1-v(x))\Delta_{n_m}(x)}{\sigma_{u}(x)} , r_{n_m}^{1/2} \frac{v(x) \Delta_{n_m}(x)}{\sigma_l(x)} 
    \end{bmatrix}^\top
    \eqqcolon B_n(t,x;v),
\end{align*}
and, since $\Delta_{n_m}(x)v(x) = \theta(x) -  \theta_{l,n_m}(x)$ and $\Delta_{n_m}(x)\{1-v(x)\} = \theta_{u,n_m}(x) - \theta(x)$, we can also regard $B_n(t,x;v)$ as the function  with respect to $\theta$
\begin{align*}
    B_n(t,x) = s_t \begin{bmatrix}
        r_{n_m}^{1/2}\frac{\theta_{u,n_m}(x) - \theta(x)}{\sigma_{u}(x)} , r_{n_m}^{1/2} \frac{\theta(x) -  \theta_{l,n_m}(x)}{\sigma_l(x)} 
    \end{bmatrix}^\top \eqqcolon B_{n_m}(t,x;\theta).
\end{align*}
Therefore, we can see that
\begin{align*}
    &  \inf_{F\in\mathcal{F}}  \inf_{\theta\in\Theta_{n_m}(F)} P_F \left[ \sup_{(t,x)\in\{u,l\}\times\mathcal{X}} (G(\tilde\psi_{t,x}) - B_{n_m}(t,x)) \le c_{n_m}^\fim \right]  \\
    & = \inf_{F\in\mathcal{F}}  \inf_{\theta\in\Theta_{n_m}(F)} P_F \left[ \sup_{(t,x)\in\{u,l\}\times\mathcal{X}} (G(\tilde\psi_{t,x}) - B_{n_m}(t,x;\theta)) \le c_{n_m}^\fim \right] .
\end{align*}
Since  $v(x) \mapsto \theta(x) = \theta_{l,n} + v(x)\Delta_n(x)$ is surjective over $\mathcal{X}\times \{l,u\}$,
\begin{align*}
    &  \inf_{F\in\mathcal{F}}  \inf_{\theta\in\Theta_{n_m}(F)} P_F \left[ \sup_{(t,x)\in\{u,l\}\times\mathcal{X}} (G(\tilde\psi_{t,x}) - B_{n_m}(t,x;\theta)) \le c_{n_m}^\fim \right] \\
    & = \inf_{F\in\mathcal{F}}  \inf_{v\in[0,1]^\mathcal{X}}P_F \left[ \sup_{(t,x)\in\{u,l\}\times\mathcal{X}} (G(\tilde\psi_{t,x}) - B_{n_m}(t,x;v)) \le c_{n_m}^\fim \right] \\
    & = \inf_{F\in\mathcal{F}} \inf_{v\in[0,1]^{\mathcal{X}}} P_F \left[  G(\tilde\psi_{t,x}) - B_{n_m}(t,x;v) \le c_{n_m}^\fim  \text{ for all } (t,x)\in \{u,l\}\times\mathcal{X} \right].
\end{align*}
Using log-concavity of the Gaussian density (cf. pp.1310 in \citealp{Stoye:2009} and pp.8 in \citealp{Frandsen_Pond:2025}) in conjunction with \cref{as:assumption4}(i) gives
\begin{align*}
    & \inf_{v\in[0,1]^{\mathcal{X}}} P_F \left[  G(\tilde\psi_{t,x}) - B_{n_m}(t,x;v) \le c_{n_m}^\fim  \text{ for all } (t,x)\in \{u,l\}\times\mathcal{X} \right] \\
    & = \inf_{v\in\{0,1\}^{\mathcal{X}}} P_F \left[  G(\tilde\psi_{t,x}) - B_{n_m}(t,x;v) \le c_{n_m}^\fim  \text{ for all } (t,x)\in \{u,l\}\times\mathcal{X} \right].
\end{align*}
Summing up, 
\begin{align*}
    & \lim_{m\to\infty}  \inf_{F\in\mathcal{F}}  \inf_{\theta\in\Theta_{n_m}(F)} P_F\left(\theta(x) \in \hat{\CI}^\fim_{n_m}(x) \text{ for all } x\in\mathcal{X} \right) \\
    &=\lim_{m\to\infty} \inf_{F\in\mathcal{F}}\inf_{v\in\{0,1\}^\mathcal{X}} P_F\left[ \sup_{(t,x)\in\{u,l\}\times\mathcal{X}} (G(\tilde\psi_{t,x}) - B_{n_m}(t,x;v)) \le c_{n_m}^\fim \right] = 1-\alpha.
\end{align*}

\paragraph{Case 2}
Next, for each $m\ge 1$ and $F\in\mathcal{F}$, define
\begin{align*}
    D_m(F) \coloneqq \{x\in\mathcal{X} : \Delta_{n_m}(x;F) > k_{n_m}\}, ~~ D_m^c(F) \coloneqq \mathcal{X}\setminus D_m(F).
\end{align*}
The proof of this case is almost same as that of \cite{Frandsen_Pond:2025}.
Since \cref{as:assumption2} assume that $k_{n_m}\sqrt{r}_{n_m} \to \infty$, it holds that 
\begin{align*}
    \inf_{x\in D_m(F)} \sqrt{r_{n_m}} \Delta_{n_m}(x) > \inf_{x\in D_m(F)} \sqrt{r_m}k_{n_m} \to \infty
\end{align*}
Also, for any $x\in D_m(F)$,  either
\begin{align*}
     \sqrt{r_{n_m}}\left( \theta(x) - \theta_{l,n_m}(x) \right) \to \infty, \text{ or } \sqrt{r_{n_m}} \left( \theta_{u,n_m}(x) - \theta(x) \right) \to \infty, \text{ or both.}
\end{align*}
If both diverge, then both $ \hat{A}_{n_m}(x) \cap \hat{B}_{n_m}(x)$ ouccurs with probability approaching to $1$.
If 
\begin{align*}
     \sqrt{r_{n_m}}\left( \theta(x) - \theta_{l,n_m}(x) \right) \to \infty, \text{ and } \sqrt{r_{n_m}} \left( \theta_{u,n_m}(x) - \theta(x) \right) < \infty,
\end{align*}
then $\hat{B}_{n_m}(x)$ occurs with probability approaching to $1$, and so convergence at the point $x$ is determined by the event $\hat{A}_{n_m}(x)$, and
\begin{align*}
    \hat{A}_{n_m}(x) \supseteq \left\{ -\frac{\hat{\sigma}_{u,n_m}(x)}{\sigma_u(x)} \hat{c}^{\fim}_{n_m} \le r_{n_m}^{1/2} \frac{\hat{\theta}_{u,n_m}(x) - \theta_{u,n_m}(x)}{\sigma_u(x)} \le \hat{c}_{n_m}^\fim + r_{n_m}^{1/2}\frac{\Delta_{n_m}(x)}{\sigma_{l}(x)}  \right\},
\end{align*}
and if
\begin{align*}
     \sqrt{r_{n_m}}\left( \theta(x) - \theta_{l,n_m}(x) \right) < \infty, \text{ and } \sqrt{r_{n_m}} \left( \theta_{u,n_m}(x) - \theta(x) \right) \to \infty,
\end{align*}
then $\hat{A}_{n_m}(x)$ occurs with probability approaching to $1$, and so convergence at the point $x$ is determined by the event $\hat{B}_{n_m}(x)$, and
\begin{align*}
    \hat{B}_{n_m}(x) \supseteq \left\{ -\hat{c}_{n_m}^\fim - r_{n_m}^{1/2}\frac{\hat{\Delta}_{n_m}(x)}{\sigma_u(x)} \le r_{n_m}^{1/2} \frac{\hat{\theta}_{l,n_m}(x) - \theta_{l,n_m}(x)}{\sigma_l(x)}  \le \frac{\hat{\sigma}_{l,n_m}(x)}{\sigma_l(x)} \hat{c}_{n_m}^\fim \right\}.
\end{align*}
Combining these, the coverage event for the point $x$ therefore includes with probability approaching to $1$ the following event, parametrized by $v(x)\in\{0,1\}$
\begin{align*}
    \hat{D}_{n_m}(x;v) \coloneqq \left\{
    \begin{array}{ll}
        \left\{ - \hat{c}_{n_m}^\fim - r_{n_m}^{1/2} \frac{\hat{\Delta}_{n_m}(x)}{\sigma_u(x)} \le r_{n_m}^{1/2} \frac{\hat{\theta}_{l,n_m}(x) - \theta_{l,n_m}(x)}{\sigma_l(x)}  \le \frac{\hat{\sigma}_{l,n_m}(x)}{\sigma_l(x)} \hat{c}_{n_m}^\fim \right\}, & v(x)=0 \\
        \left\{-\frac{\hat{\sigma}_{u,n_m}(x)}{\sigma_u(x)} \hat{c}^{\fim}_{n_m} \le r_{n_m}^{1/2} \frac{\hat{\theta}_{u,n_m}(x) - \theta_{u,n_m}(x)}{\sigma_u(x)} \le \hat{c}_{n_m}^\fim + r_{n_m}^{1/2}\frac{\Delta_{n_m}(x)}{\sigma_{l}(x)} \right\}, & v(x)=1 
    \end{array}\right. ~~~,
\end{align*}
where the parameter $v(x)$ corresponds to whether $\sqrt{r_{n_m}}(\theta(x) - \theta_{l,n_m}(x))$ or $\sqrt{r_{n_m}} \left( \theta_{u,n_m}(x) - \theta(x) \right)$ diverges.
Therefore, we have 
\begin{align*}
    & P_F\left(\theta(x) \in \hat{\CI}_{n_m}^\fim(x)\right) \\
    & \ge P_F\left( \left\{ \hat{A}_{n_m}(x) \cap \hat{B}_{n_m}(x) \text{ for all } x\in D^c_m(F) \right\} \cap \left\{\hat{D}_{n_m}(x;v) \text{ for all } x\in D_m(F) \right\} \right) + o(1).
\end{align*}
Define a counterpart to $\hat{D}_{n_m}(x;v)$ as
\begin{align*}
    \tilde{D}_{n_m}(x;v)\coloneqq \left\{
    \begin{array}{ll}
        \left\{-c_{n_m}^\fim - r_{n_m}^{1/2} \frac{{\Delta}_{n_m}(x)}{\sigma_u(x)} \le r_{n_m}^{1/2} \frac{\hat{\theta}_{l,n_m}(x) - \theta_{l,n_m}(x)}{\sigma_l(x)}  \le {c}_{n_m}^\fim \right\}, & v(x)=0 \\
        \left\{-{c}^{\fim}_{n_m} \le r_{n_m}^{1/2} \frac{\hat{\theta}_{u,n_m}(x) - \theta_{u,n_m}(x)}{\sigma_u(x)} \le {c}_{n_m}^\fim + r_{n_m}^{1/2}\frac{\Delta_{n_m}(x)}{\sigma_{l}(x)} \right\}, & v(x)=1 
    \end{array}\right. ~~~,
\end{align*}
From  \eqref{eq:variance_consistency}, $|\hat{c}_{n_m} - c_{n_m}| = o_p(1)$ (\cref{as:assumption3}) and the fact that
\begin{align*}
    \sqrt{r_{n_m}} \Delta_{n_m}(x) > \sqrt{r_{n_m}}k_{n_m} \to \infty \text{ on } D_m(F),
\end{align*}
we have
\begin{align*}
    & P_F\left( \left\{ \hat{A}_{n_m}(x) \cap \hat{B}_{n_m}(x) \text{ for all } x\in D^c_m(F) \right\} \cap \left\{\hat{D}_{n_m}(x;v) \text{ for all } x\in D_m(F) \right\} \right) \\
    & = P_F\left( \left\{ \tilde{A}_{n_m}(x) \cap \tilde{B}_{n_m}(x) \text{ for all } x\in D^c_m(F) \right\} \cap \left\{\tilde{D}_{n_m}(x;v) \text{ for all } x\in D_m(F) \right\} \right) + o(1).
\end{align*}
Define
\begin{align*}
    & I_1(v) \coloneqq \left\{ (t,x) \in \{u,l\}\times D_m^c(F) \mid v(x)\in[0,1]\right\}, \\
    & I_2(v) \coloneqq \{(t,x) \in \{l\} \times D_m(F) \mid v(x) =0\}, \quad  I_3(v) \coloneqq \{(t,x) \in \{u\} \times D_m(F) \mid v(x) =1 \},
\end{align*}
and $I(v) \coloneqq I_1(v) \cup I_2(v) \cup I_3(v)$, then in the same way as the proof of \cref{lem:uniform_Gauss}, we have
\begin{align*}
    & P_F\left( \left\{ \tilde{A}_{n_m}(x) \cap \tilde{B}_{n_m}(x) \text{ for all } x\in D^c_m(F) \right\} \cap \left\{\tilde{D}_{n_m}(x;v) \text{ for all } x\in D_m(F) \right\} \right) \\
    & = P_F\left[ \sup_{(t,x)\in I(v)} \left(  G_{n_m}(t,x) - B_{n_m}(t,x;v)\right) \le c_{n_m}^\fim \right] + o(1),
\end{align*}
then under \cref{as:assumption1}, \cref{as:assumption4} and \cref{as:assumption5}, also in the same way as the proof of \cref{lem:uniform_Gauss}, we can see that
\begin{align*}
    & \sup_{F\in\mathcal{F}} \left|  P_F\left[ \sup_{(t,x)\in I(v)} \left(  G_{n_m}(t,x) - B_{n_m}(t,x;v)\right) \le c_{n_m}^\fim\right] \right. \\
    & \left. \quad\quad\quad\quad\quad - P_F\left[ \sup_{(t,x)\in I(v)} \left(  G(t,x) - B_{n_m}(t,x;v)\right) \le c_{n_m}^\fim\right] \right| = o(1).
\end{align*}
Since $I(v)\subseteq \{u,l\}\times\mathcal{X}$, we have
\begin{align*}
    & P_F\left[ \sup_{(t,x)\in I(v)} \left(  G(t,x) - B_{n_m}(t,x;v)\right) \le c_{n_m}^\fim\right]  \\
    & \ge P_F\left[ \sup_{(t,x)\in \{u,l\}\times\mathcal{X}} \left(  G(t,x) - B_{n_m}(t,x;v)\right) \le c_{n_m}^\fim\right].
\end{align*}
Therefore
\begin{align*}
     & P_F\left(\theta(x) \in \hat{\CI}_{n_m}^\fim(x)\right) \ge P_F\left[ \sup_{(t,x)\in \{u,l\}\times\mathcal{X}} \left(  G(t,x) - B_{n_m}(t,x;v)\right) \le c_{n_m}^\fim\right] + o(1).
\end{align*}
Taking the infmum over $F\in\mathcal{F}$ and $\theta\in \Theta_{n_m}(F)$ and similar argument in Case 1 give
\begin{align*}
    & \inf_{F\in\mathcal{F}} \inf_{\theta\in \Theta_{n_m}(F) } P_F\left(\theta(x) \in \hat{\CI}_{n_m}^\fim(x)\right) \\
    & \ge \inf_{F\in\mathcal{F}} \inf_{v\in\{0,1\}^\mathcal{X}}  P_F\left[ \sup_{(t,x)\in \{u,l\}\times\mathcal{X}} \left(  G(t,x) - B_{n_m}(t,x;v)\right) \le c_{n_m}^\fim\right] + o(1)
\end{align*}
Therefore
\begin{align*}
    & \liminf_{m\to\infty}\inf_{F\in\mathcal{F}} \inf_{\theta\in \Theta_{n_m}(F) } P_F\left(\theta(x) \in \hat{\CI}_{n_m}^\fim(x)\right) \\
    & \ge \liminf_{m\to\infty}\inf_{F\in\mathcal{F}} \inf_{v\in\{0,1\}^\mathcal{X}}  P_F\left[ \sup_{(t,x)\in \{u,l\}\times\mathcal{X}} \left(  G(t,x) - B_{n_m}(t,x;v)\right) \le c_{n_m}^\fim\right] + o(1)\\
    & = 1-\alpha.
\end{align*}
Thus, the same argument in  the last paragraph of the proof of Proposition 1 in \cite{Stoye:2009} gives
\begin{align*} 
    \lim_{m\to\infty}  \inf_{F\in\mathcal{F}}  \inf_{\theta_{n_m}\in\Theta_{n_m}(F)} P_F\left( \theta_{n_m}(x)\in \hat{\CI}_{n_m}^\fim(x)  \text{ for all } x\in\mathcal{X} \right) = 1-\alpha.
\end{align*}
\end{proof}

\subsection{Proofs for \cref{subsec:implementation}}
\begin{proof}[Proof of \cref{thm:vaild_mb}]
First, from \cref{as:assumption6}(i), it holds that
\begin{align}
    & \sup_{F\in\mathcal{F}} P_F\Bigg(  \Bigg|\sup_{(t,x)\in\{u,l\}\times\mathcal{X}} \left( \mathbb{G}_n^{\star,b}(\tilde{\psi}^{\text{s.a.}}_{t,x}) - \hat{B}_n(t,x;v) \right) \nonumber\\
    & \qquad\qquad- \sup_{(t,x)\in\{u,l\}\times\mathcal{X}} \left( \mathbb{G}_n^{\star,b}(\tilde{\psi}_{t,x}) - {B}_n(t,x;v) \right) \Bigg| \ge \varepsilon \Bigg) = o(1), \label{eq:sa-consistency}
\end{align}
for all $\varepsilon>0$.
Also, for each $F\in\mathcal{F}$ and every $\gamma_n\in(0,1)$,  under \cref{as:assumption4} and \cref{as:assumption6}(iii), Theorem 2.2 in \cite{chernozhukov2016empirical} gives
\begin{align*}
    P_F\left(\left|\sup_{\tilde\psi_{t,x}\in\tilde\Psi} \left( \mathbb{G}_n^{\star,b}(\tilde{\psi}_{t,x}) - B(\tilde\psi_{t,x}) \right) - \sup_{\tilde\psi_{t,x}\in\tilde\Psi}(G(\tilde\psi_{t,x}) - B(\tilde\psi_{t,x})) \right| > r_1^\star(F)\right) \le r_2^\star.
\end{align*}
with $r_2^\star \coloneqq C_4(\gamma_n + n^{-1})$,  with a slight abuse of notation, the Gaussian supremum is understood as a coupled copy having the same distribution as $\sup_{\tilde\psi_{t,x}\in\tilde\Psi}(G(\tilde\psi_{t,x}) - B(\tilde\psi_{t,x}))$. 
Therefore, by Markov’s inequality, for any sequence $\rho_n \to 0$ satisfying $r_2^\star/\rho_n \to 0$, we have
\begin{align*}
    P^\star\left(  \left|\sup_{\tilde\psi_{t,x}\in\tilde\Psi} \left( \mathbb{G}_n^{\star,b}(\tilde{\psi}_{t,x}) - B(\tilde\psi_{t,x}) \right) - \sup_{\tilde\psi_{t,x}\in\tilde\Psi}(G(\tilde\psi_{t,x}) - B(\tilde\psi_{t,x})) \right| > r_1^\star(F) \right) \le \rho_n,
\end{align*}
with probability approaching one uniformly over $F\in\mathcal{F}$, where $P^\star$ is probability conditional on observations.
As in the case of the proof of \cref{lem:uniform_Gauss}, Lemma 2.1 and Lemma 2.2 in \cite{chernozhukov2016empirical} (with appropriate choice of $r$ and $\delta$ so that $r\delta\lesssim r_1^\star(F)$ and $\phi(\delta)\lesssim r_1^\star(F)$ in their notation), we have
\begin{align*}
    & \sup_{F\in\mathcal{F}} P_F \Bigg(   \sup_{z\in\mathbb{R}}  \Bigg|P^\star\left[\sup_{\tilde\psi_{t,x}\in\tilde\Psi} \left(  \mathbb{G}_n^{\star,b}(\tilde\psi_{t,x}) - B(\tilde\psi_{t,x})\right) \le z \right]\\
    & \qquad\qquad\qquad  -  P_F\left[ \sup_{\tilde\psi_{t,x}\in\tilde\Psi} (G(\tilde\psi_{t,x}) - B(\tilde\psi_{t,x})) \le z \right]\Bigg| \ge \varepsilon\Bigg)\\
    & \le \sup_{F\in\mathcal{F}} P_F \left(  \sup_{z\in\mathbb{R}}  P_F\left[ \left|\sup_{\tilde\psi_{t,x}\in\tilde\Psi} (G(\tilde\psi_{t,x}) - B(\tilde\psi_{t,x})) - z\right| \le r_1^\star(F) \right] + \rho_n \ge \varepsilon \right) + \frac{r_2^\star}{\rho_n}  \\
    & \lesssim \sup_{F\in\mathcal{F}}  P_F \left( \frac{2}{\sqrt{\inf_{\tilde\psi_{t,x}\in\tilde\Psi} \Var_F\{G(\tilde\psi_{t,x})\}}} r_1^\star(F) \left(1 + \sqrt{K_n(F)}\right) + \rho_n  \ge \varepsilon \right) + \frac{r_2^\star}{\rho_n}.
\end{align*}
This result, \cref{eq:sa-consistency}, and Assumption \cref{as:assumption6}(ii) give
\begin{align*}
    & \sup_{F\in\mathcal{F}} P_F\Bigg(\sup_{z\in\mathbb{R}}\left| P^\star\left[ \sup_{(t,x)\in\{u,l\}\times\mathcal{X}} \left( \mathbb{G}_n^{\star,b}(\tilde{\psi}^{\text{s.a.}}_{t,x}) - \hat{B}_n(t,x;v) \right) \le z \right]\right. \\
    &\quad\quad\quad\quad\quad\quad\quad\quad - \left.  P_F\left[ \sup_{\tilde\psi_{t,x}\in\tilde\Psi} (G(\tilde\psi_{t,x}) - B(\tilde\psi_{t,x})) \le z \right]  \right| \ge \varepsilon \Bigg) = o(1).
\end{align*}

\end{proof}

\bibliography{refs}
\bibliographystyle{apalike}
\end{document}